%% file: main.tex
\pdfoutput=1
\documentclass[a4paper,UKenglish,cleveref, autoref, thm-restate]{lipics-v2021}
\hideLIPIcs
\nolinenumbers
\usepackage[utf8]{inputenc}
\usepackage{amsmath}
\usepackage{amssymb}
\usepackage{amsthm}
\usepackage{amsfonts}
\usepackage{xspace}
\usepackage{enumerate}
\usepackage{comment}
\usepackage{mycommands} 
\newcounter{enumi_val}
\title{On polynomial grammars extended with substitution}

\author{Janusz Schmude}{Institute of Informatics,\\ University of
	Warsaw, Poland}{jschmude@mimuw.edu.pl}{}{}

\authorrunning{J. Schmude}

\Copyright{Janusz Schmude}

\ccsdesc[100]{Theory of computation~Formal languages and automata}

\keywords{equivalence problem, register transducers, word substitution, polynomial grammar}

\funding{Supported by ERC Consolidator Grant LIPA, agreement no. 683080.}

\acknowledgements{I would like to thank Miko\l{}aj Boja\'{n}czyk for the support in obtaining the results and writing this paper. I would also like to thank James Worrell, Mahsa Shirmohammadi, Sandra Kiefer, S\l{}awomir Lasota and Rados\l{}aw Pi{\'o}rkowski for helpful discussions.}
\begin{document}

	\maketitle
	
	\begin{abstract}
	\input{abstract}
	\end{abstract}
	\section{Introduction}
	\input{intro}
	\section{Preliminaries}\label{sec:prelim}
	\input{prelim}
	\subsection{Fields, algebraic sets}
	\input{alg-geo-prelim}

	\section{Polynomial grammars with a bounded number of substitutions}\label{sec:modeltwo}
	\input{modeltwo}
	\section{\PGevals}\label{sec:modelone}
	\input{modelone}
	\section{Summary, future work}\label{sec:sum}
	\input{summary}
	%
	\bibliographystyle{plain}
	\bibliography{biblio}
	\appendix\section{Appendix}
	\input{appendix}
\end{document}

%% file: abstract.tex
We investigate decidability of equivalence of register transducers, also called copyful Streaming String Transducers in case of string input, extended with an operation of substituting a register for all occurrences of a given letter in another register. We reduce to zeroness of polynomial grammars (over ring of polynomials) extended with analogous substitution operation by encoding strings into polynomials; a similar method was used successfully by Seidl et al. in 2018.
We give two restrictions under which register transducers with substitution have decidable equivalence. They seem to be very restrictive but on the other hand, they seem to be on the edge of the scope of this ``polynomial'' method, as in the third result we give a rather restricted model of polynomial grammars with substitution that has undecidable equivalence.

%% file: intro.tex
Let us consider the following computation device that computes string-to-string functions: it has two registers $R, S$ that store strings and reads input words from left to right. It initializes the registers to empty words and if the read letter is $a$, it updates the registers by putting simultaneously
$$
\begin{cases}
	R := a\cdot R,\\ 
	 S := R \cdot a.
\end{cases}
$$
After input word is read, it outputs $R\cdot S$. For example, a run over input word $w=abb$ is 
$$
\!\!
\begin{tabular}{c}$R = \varepsilon$\\ $S = \varepsilon $\end{tabular}\!\!\xra{a}\!\! 
\begin{tabular}{c}$R = a$\\ $ S = a $\end{tabular}\!\!\xra{b} \!\!
\begin{tabular}{c}$R = ba$\\ $ S = ab$\end{tabular}\!\!\xra{b} \!\!
\begin{tabular}{c}$R = bba$\\ $ S = abb $\end{tabular}\!\!\ra
 \text{output }  bbaabb.
  $$
This device computes function $w \mapsto \rev(w) \cdot w$, where $\rev$ is the reverse function. 
Such a device is an instance of a \emph{register transducer} \cite{BojSiglog2019} (in general, register transducers can compute tree-to-string functions, however in all the examples we give string-to-string functions, which are a special case where input trees are monadic). Sometimes two register transducers defined in two different ways compute the same function -- the \emph{equivalence problem} asks, given two register transducers, if this is the case. For example, the register transducers that compute, respectively, the reverse and identity function are equivalent if and only if the alphabet is unary.

There are several methods of deciding equivalence of register transducers e.g. by use of an equivalent model of two-way transducers \cite{Gurari82} or by reduction to non-emptiness of one counter automata in case of \emph{copyless} register transducers, which are called \emph{Streaming String Transducers} \cite{Alur11}.
Method that is close to one we use in this paper is by reduction to \emph{zeroness problem} for \emph{polynomial grammars} \cite{BojSiglog2019}. 

Let us describe a variant of polynomial grammars that uses a free monoid instead of a ring and explain how it is related to equivalence problem of register transducers.
It is a generalisation of context-free grammar, in which each nonterminal is associated with some dimension $k$, and therefore outputs $k$-tuples, and right-hand sides of production rules may use function symbols to avoid introducing independent copies of a nonterminal. For example, consider a \pg over $(\Sigma^*, \cdot)$ with one nonterminal $X$ of dimension 2 and production rules 
 \begin{align*}
 	X \ra  \ p_a(X), & \ a \in \Sigma \mid (\varepsilon, \varepsilon), \text{where} \\
 	p_a:(\Sigma^*)^2 \ra& (\Sigma^*)^2 \\
 	p_a(w_1, w_2) =& \ (aw_1, w_2a) \text{ for } a \in \Sigma, w_1, w_2 \in \Sigma^*. 
 \end{align*}
By definition, its language consists of pairs $p_{a_n}(p_{a_{n-1}}(\ldots (p_{a_2}(p_{a_1}(\varepsilon, \varepsilon)))\ldots))$ for $a_1, \ldots, a_n \in \Sigma$. We define the \emph{grammar equivalence problem}, in the case the initial nonterminal is 2-dimensional, as follows: is the first coordinate equal to the second coordinate for every produced pair? Observe that equivalence holds for the mentioned grammar if and only if the register transducers that compute, respectively, reverse and identity function, are equivalent.

Register transducers considered in this paper are enriched with a substitution operation, we call them \emph{register transducers with substitution}. For example, consider the function $\sqrev:\Sigma^* \ra (\Sigma \cup \{\#\})^*$, called \emph{squared reverse} (the definition and name are motivated by \emph{iterated reverse} \cite{BojPolyregularArXiv2018}), defined as follows:
$$
	\sqrev : w \mapsto (\# \cdot \rev(w))^{|w|}.
$$
It can be computed by two following register transducers with substitution that initialize their registers on empty strings, output register $R$, and have register updates as follows: 
$$
\begin{cases}
R := \# \cdot a \cdot S \cdot (R[\# := \# \cdot a]), \\
S := a\cdot S,
\end{cases}
\text{ for } a \in \Sigma
$$ 
and 
$$
\begin{cases}
	R :=  (R[\# := \# \cdot a]) \cdot \# \cdot a \cdot S , \\
	S := a \cdot S,
\end{cases}
\text{ for } a \in \Sigma.
$$ 
To decide their equivalence, one may consider a \emph{polynomial grammar with substitution} (again, over free monoid instead of a ring) with initial nonterminal $Y$ of dimension 4: 
$Y \ra q_a(Y), a \in \Sigma \mid (\varepsilon, \varepsilon, \varepsilon, \varepsilon)$, where $q_a(w_1, w_2, w_3, w_4) = 
(\# \cdot a\cdot w_2\cdot(w_1[\# := \# \cdot a]), a\cdot w_2, (w_3[\# := \# \cdot a])\cdot\# \cdot a\cdot w_4, a\cdot w_4)$ and test for equivalence the grammar $S_1 \ra \pi_{[1,3]}(Y)$, where $\pi_I$ is projection on coordinates from $I$.

We are interested in decidability of grammar equivalence for such grammars.

We use the Hilbert Method \cite{BojSiglog2019}, which relies on encoding objects manipulated by transducers, like strings or trees, into rings e.g. of integers or of polynomials; this is done in order to use algebraic geometry machinery. This method has been used to prove decidability of equivalence for register transducers without substitution, where strings were encoded into ring of integers \cite{SeidlEtAl2018} (see also \cite{BojSiglog2019} for a slightly different encoding).
In this paper we use an analogous encoding (Definition \ref{df:str-to-pol-encoding}), this time into ring of polynomials, in order to cover substitution, thus reducing problem to equivalence of \emph{\pgsubss} over a ring of polynomials. 
Equivalence is known to be undecidable for this model, even in a rather restricted special case \cite[Theorem 20]{BoiretEtAl2018}. In this paper we prove two decidability results for other special cases and strengthen this undecidability result, clarifying the scope of this ``polynomial'' approach.

Obtained decidability results transfer immediately to corresponding register transducers. This is stated in Lemma \ref{lm:RTsubss-can-be-simulated-by-pgsubss} for unrestricted \pgsubss, which is a reduction to an undecidable problem, however, as explained in Remark \ref{rm:reduction-to-an-undecidable-problem}, reduction holds also for special cases, which are decidable (Theorem \ref{thm:modeltwo-decidable} and Lemma \ref{lm:pgaut-decidable}).

\partofsec{Main results.}
The main problem discussed in this paper is \emph{zeroness} of \pgsubss over a ring of polynomials (Definition \ref{df:zeroness}).
First we discuss such grammars with an additional fixed bound on the number of transitions per derivation that use substitution. We provide a positive result (Theorem \ref{thm:modeltwo-decidable}) and a negative one (Theorem \ref{thm:modeltwoundecidable-undecidable}).
The negative result is a strengthening of  \cite[Theorem 20]{BoiretEtAl2018}, mentioned in the previous paragraph. 
Then we move to grammars that use a special case of substitution -- evaluation. For this variant we give a positive result (Lemma \ref{lm:pgaut-decidable}) and describe how it transfers to register transducers with substitution (Theorem \ref{thm:modelone-decidable}).

%


%% file: prelim.tex
\partofsec{Polynomial functions.}
By a \emph{ring} we mean a commutative ring with unity. A ring has \emph{no zero divisors} if for all non-zero $a,b\in R$, their multiplication $a\cdot b$ is non-zero. Let $R$ be a ring, for example the ring of integers. Following \cite{BojSiglog2019}, it is convenient to define polynomial functions so that they can return tuples. Tuples inputted by polynomials will be also called \emph{vectors}. 
A function $$p:R^n \ra R^k$$ for some $n, k \geq 0$ is a \emph{polynomial function} if every output coordinate is represented by a polynomial in $n$ variables. For a set $V \subseteq R^n$, by $p\restr{V}$ we denote function $p$ restricted to $V$.

\partofsec{Substitution.}
In the ring $R[X]$, substitution is a $(|X|+1)$ - ary operation defined as 
\begin{align*}
	\subs&: R[X] \times R[X]^{X} \ra R[X]\\
	\subs& (p, (q_x)_{x \in X}) = p[x:=(q_x)_{x \in X}].
\end{align*}
We consider a vectorized variant of substitution:
\begin{align*}
\subs^{Y}_{X}&: R[X]^{Y} \times R[X]^{X} \ra R[X]^{Y}\\
\subs^{Y}_{X}&((p_y)_{y \in Y}, (q_x)_{x \in X}) = (p_y(x:=q_x, x\in X))_{y \in Y}.
\end{align*}
If the arities are clear from the context,  we simply write
$
\subs(\w p, \w q),
$
or even
$$
\w p (\w q),
$$
where $\w p = (p_y)_{y \in Y}$ and $\w s = (q_x)_{x \in X}$ are tuples of polynomials in $R[X]$.

\partofsec{\PGs.}
We define polynomial grammars following \cite[Definition 2.1]{BojSiglog2019}. As opposed to the Introduction, we define them only for rings. Most of the times it will be the ring of polynomials, although in the proof of Theorem \ref{thm:modeltwo-decidable} we also use other rings.
\begin{definition}\label{df:pol-grammar}
A \emph{polynomial grammar} over ring $R$ consists of
\begin{itemize}
	\item finite set of \emph{nonterminals} with a distinguished \emph{initial nonterminal}, each nonterminal with an associated \emph{dimension} $\in \{0, 1, 2, \ldots\}$, and
	\item finite set of \emph{production rules}, each of form $X \ra p(Y_1, \ldots, Y_k)$ for $k=0,1,\ldots$ where $p: R^n \ra R^m$ is a polynomial function, $Y_i$ for $i=1,\ldots,k$ are nonterminals, and input and output dimensions of expressions ``match'', i.e. sum of dimensions of $Y_i$'s, for $i=1\ldots l$, is equal to $n$, and dimension of $X$ is equal to $m$.

\end{itemize}
If a nonterminal has dimension $n$, then it generates a subset of $R^n$, which is defined as follows by induction. (The language generated by the grammar is defined to be the subset generated by its initial nonterminal.) Suppose that
$$
X \ra p(Y_1,...,Y_k)
$$
is a production and we already know that vectors $v_1 , \ldots , v_k$ are generated by nonterminals $Y_1,...,Y_k$ respectively. Then the vector $p(v_1 ,\ldots , v_k)$ is generated by nonterminal $X$. The induction base is the special case of $k = 0$, where the polynomial $p$ is a constant.

\begin{ex}[see Fig. 1]
	Polynomial grammar $$A\ra p(A) \mid 2$$ with $p(a) = a \cdot a$ has language $$\{2^{2^n}: n \in \Z_+\}.$$
	A similar polynomial grammar, $$A \ra A \cdot A \mid 2$$ has language $$\{2^n: n \in \Z_+\}.$$
\end{ex}

We define the \emph{dimension} of a grammar to be the dimension of its initial nonterminal. 
For simplicity, the language of a grammar $G$ with initial nonterminal $A$ is sometimes denoted by $A$.
For a polynomial function $f:R^n \ra R^m$, for a grammar $G$ with initial nonterminal $A$ of dimension $n$, by $f(G)$ or $f(A)$ we denote a ``canonical'' grammar with language $f(A)$, i.e. grammar build from $G$ by adding a fresh initial nonterminal $S$ and a production $S \ra f(A)$.
\end{definition}
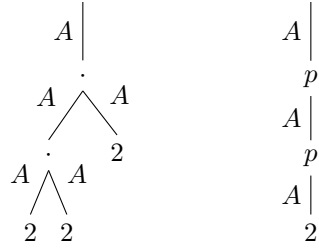
\begin{figure} \label{fig:ex-derivation}
	\centering
	\begin{tikzpicture}
		\Tree[ \edge node[auto=right]{$A$};
			[.$\cdot$
			\edge node[auto=right]{$A$};
				[.$\cdot$
				\edge node[auto=right]{$A$};
					[.2
					]
				\edge node[auto=left]{$A$};
					[.2
					]
				]	
			\edge node[auto=left]{$A$};
				[.2
				]
			]	
		]
		\begin{scope}[xshift=3cm]
			\Tree[
			\edge node[auto=right]{$A$};
			[.$p$ 
			\edge node[auto=right]{$A$};
			[.$p$
			\edge node[auto=right]{$A$};
			[.2
			]
			]
			]
			]
		\end{scope}
	\end{tikzpicture}
\caption{Example derivation trees for grammars $A\ra p(A) \mid 2$ with $p(a) = a \cdot a$ and $A \ra A \cdot A \mid 2$. They produce $8 = 2^3$ and, respectively, $16=2^{2^2}$.}
\end{figure}

\partofsec{\PGsubss.}
Let $R$ be a ring. A \emph{\pgsubs} over $R[X]$ is an extension of \pg over $R[X]$ where polynomials in production rules may use substitution (additionally to $+$ and $\cdot$).
\begin{definition}\label{df:zeroness}
	For a \pg or a \pgsubs $A$ over a ring of polynomials, that has a non-empty language
	\footnote{A technical assumption, does not change the essence of the problem. Notice it can be checked easily - it is enough to find productive nonterminals e.g. by fixpoint strategy. In particular, the algorithm does not depend on the output algebra.},
	 \emph{zeroness problem} asks if $L(A) = \{0\}$; we often abbreviate it to $A=0$.
\end{definition}

\partofsec{String-to-polynomial encoding.} 
The principal application of zeroness problem is equivalence of register transducers; the latter can be reduced to the former by the encoding we describe now.
\begin{definition}\label{df:str-to-pol-encoding}
Consider the function $\Phi$ that maps strings over $\Sigma$ into pairs of polynomials in variables from set $\wt \Sigma \cup \w \Sigma$, which are two disjoint copies of $\Sigma$, defined as follows:
\begin{align}\label{eq:str-to-pol-enc}
	\Phi&: a_1\ldots a_n \mapsto \left( \sum_{i=1}^n \w{a_1} \cdot \ldots \cdot\w{a_{i-1}} \wt {a_i}, 
	\quad
	 \w{a_1} \cdot \ldots  \w{a_n}\right).
\end{align}
We denote first and second coordinate of  $\Phi(w)$ by $\wt w$ and $\w w$, i.e.
\begin{equation}
	\Phi(w) = (\wt w, \w w).
\end{equation}
For example, $\Phi(abbab) = (\wt a {\w b}^2 {\w a} + \wt b {\w b}^2 {\w a} + \wt b {\w b}{\w a} + \wt a{\w b} + \wt b, {\w a}^2 {\w b}^3)$.

The function $\Phi$ is a ``generalisation\footnote{It is a generalisation in a sense, that natural encoding is an evaluation of $\Phi$ at $\wt a = \tau(a), \w a = |\Sigma|$ for a bijection $\tau$ of $\Sigma$ with $\{0, 1, \ldots, |\Sigma|-1\}$}'' of a natural encoding of strings into numbers they represent in $|\Sigma|$ -ary, i.e. $w \mapsto (\val_{|\Sigma|}(w), |\Sigma|^{|w|})$ where $\val_{|\Sigma|}(w)$ is the value of $w$ as a number in $|\Sigma|$-ary (assuming some bijection of $\Sigma$ with $\{0, 1, \ldots, |\Sigma|-1\}$). We use this generalisation to be able to use substitution.

Notice that $\Phi$ ``commutes'' with substitution , i.e.  
$$ 
\Phi(w[a := v_{a}, a \in \Sigma]) =
\Phi(w)[\wt a := \wt{v_{a}}, \w a  := \w {v_{a}}, a \in \Sigma] \text{ for } w, (v_a)_{a \in \Sigma} \in \Sigma^*.
$$
We say that the polynomial substitution $(\wt a := \wt{v_a}, \w a := \w{v_a})_{a \in \Sigma}$ is \emph{induced by} the word substitution $(a := v_{a})_{a \in \Sigma}$.
\end{definition}
Using the encoding above, we obtain the following lemma.
\begin{lemma}\label{lm:RTsubss-can-be-simulated-by-pgsubss}
	Equivalence of register transducers with substitution over $\Sigma^*$ can be reduced to zeroness of polynomial grammars with substitution over $\QSigR$.
\end{lemma}
\begin{remark}[Undecidability in Lemma \ref{lm:RTsubss-can-be-simulated-by-pgsubss}]\label{rm:reduction-to-an-undecidable-problem}
Let us emphasize that even though in Lemma \ref{lm:RTsubss-can-be-simulated-by-pgsubss} we have a reduction to an undecidable problem, it remains valid for special cases of this problem that are decidable -- Theorem \ref{thm:modeltwo-decidable} and Lemma \ref{lm:pgaut-decidable}. As corollaries we obtain, respectively, result from beginning of Section \ref{subsec:interpretation-as-sys-of-eqs-testing} and Theorem \ref{thm:modelone-decidable}.
\end{remark}
\begin{proof}[Proof sketch of Lemma \ref{lm:RTsubss-can-be-simulated-by-pgsubss}]
	Suppose that we have two register transducers with substitution that compute string-to-string functions $f, g: \Sigma^* \ra \Sigma^*$. We can convert them into a polynomial grammar $G$ over $\Q[\wt \Sigma \cup\w \Sigma]$ that produces pairs of polynomials $(\wt{f(w)}, \wt{g(w)})$ for input words $w$. It can be done us using the fact that
	$$
		\Phi(w \cdot v) = p(\Phi(w), \Phi(v)) 
	$$
for some polynomial function $p$ -- indeed, for any words $w, v$ we have $\wt {wv} = \wt w \cdot \w v + \wt v$ and $\w {wv} = \w w \cdot  \w v$ or explicitly 
$$p(x_1, x_2, y_1, y_2) = (x_1y_2 + y_1 , x_2y_2)$$
-- and that  $\Phi$ ``commutes'' with substitution (see Definition \ref{df:str-to-pol-encoding}).
 Applying polynomial function $p(x_1, x_2) = x_1 - x_2$ to $G$ results in grammar $p(G)$ for which zeroness holds if and only if $f, g$ are equivalent.
\end{proof}

%% file: alg-geo-prelim.tex
A \emph{field} is a ring in which every non-zero element has an inverse. Let $K$ be a field. By $K[X]$ we denote the ring of polynomials over set of variables $X$. By $K(X)$ we denote the field of rational functions over $X$, i.e. $K(X) = \{\frac{p}{q} \ | \ p, q \in K[X], q\neq 0\}$. 

A field is a \emph{computable field} if its elements can be enumerated such that operations $+, \cdot$ are computable functions.

A field $K$ is \emph{algebraically closed}, if every univariate polynomial with coefficients in $K$ has a root in $K$. An \emph{algebraic closure} of a field $K$, denoted by \barK, is, roughly speaking, the smallest algebraically closed field that contains $K$. Typical examples of algebraically closed fields are (1) the field of complex numbers, which is algebraic closure of field of real numbers and (2) the field of algebraic numbers, which is algebraic closure of field $\Q$.

Introducing algebraic closures does not affect effectiveness:
\begin{lemma}(\cite[Theorem 7]{Rabin68})\label{fact:rabin-alg-clos-computable}
	Let $K$ be a given computable field. Then one can compute an embedding $K \subset \barK$, which moreover is a computable function. In particular, \barK is a computable field and it can be computed, given $K$.
\end{lemma}
By an \emph{ideal} we mean a subset $I$ of $K[X]$ in which $f, g \in I$ implies $f + g \in I$ and $f \in I, p \in K[X]$ implies $f \cdot p \in I$, for all polynomials $f, g, p$. For a set of polynomials $F$, by $\idGenBy{F}$ we denote set $\{r_1 f_1 + r_2 f_2 + \ldots + r_k f_k \mid f_i \in F, r_i \in K[X], k \geq 0\}$; this is the smallest ideal that contains set $F$ called \emph{ideal generated by $F$}; elements of $F$ are called \emph{generators} of $\idGenBy{F}$.
Ideals can be represented in a finite way, as shown by Hilbert's Basis Theorem.
\begin{fact}(Hilbert's Basis Theorem, \cite[Chapter 2, §5, Theorem 5]{CoxOShea2015})\label{fact:hilbert-basis-theorem}.
	Let $K$ be a field. Every ideal in $K[X]$ can be generated by a finite set of polynomials.
\end{fact}

Let $K$ be a field. An \emph{algebraic set} is a subset of $K^n$ which is a set of zeros of some set of polynomials in $n$ variables, i.e. is of form $\{v \in K^n \ | \ p(v) = 0 \text{ for all } p\in F\}$ for some set of polynomials $F \subseteq K[X]$.
For an algebraic set $V$, by $I(V)$ we denote the set of all polynomials that are zero on $V$; this is an ideal.  For a ideal $I$ in $K[X]$, by $V(I)$ we denote the set of common zeros of all polynomials of $I$; this is a variety. We say that $V(I)$ is \emph{represented} by $I$. Every algebraic set can be represented by some ideal, for example by $I(V)$; this, together with Hilbert's Basis Theorem (Fact \ref{fact:hilbert-basis-theorem}), provides a way to represent algebraic sets in a finite way.

From this we conclude the following.
\begin{cor}\label{cor:algebraic-sets-enumerate}
	Let $K$ be a computable field. Then algebraic sets can be enumerated.
\end{cor}
A \emph{coordinate ring} of an algebraic set $V \subseteq K^n$ is the ring of polynomial functions from $V$ to $K$.
An algebraic set is called an (irreducible) \emph{variety} if it cannot be represented as a finite union of pairwise distinct algebraic sets.
\begin{remark}[Notation: does algebraic set = variety?]
	Different sources introduce different names for algebraic sets. For example, in \cite{CoxOShea2015} they are called varieties, and what we call a variety is called an irreducible variety. To avoid confusion, we always precede word ``variety'' with ``(irreducible)''.
\end{remark}
\begin{lemma}(\cite[Chapter 4, §6, Theorem 2 and remarks at the end of Chapter 4, §6 ]{CoxOShea2015})\label{lemma:deco-into-irreducibles}
	Let $K$ be a computable field. Then every algebraic set can be effectively decomposed into (irreducible) varieties.
\end{lemma}
\begin{lemma}(\cite[Chapter 4, §5, Proposition 3]{CoxOShea2015}) \label{lm:irreducible-iff-coordinate-ring-no-zero-divisors}
	An algebraic set is an (irreducible) variety if and only if its coordinate ring has no zero divisors.
\end{lemma}
%

%
%

%% file: modeltwo.tex
In this section we discuss decidability of zeroness problem for \pgsubss over a ring of polynomials with the following restriction: substitution operation can be used once per derivation. We call this model \emph\modeltwoundecidable.
%
Let us recall that zeroness is decidable in case of no substitutions.
\begin{lm}[\cite{BojSch2020}, Theorem 15]\label{lm:pg-over-K}
	Let $R$ be a computable ring with no zero divisors. Then equivalence is decidable for \pgs over $R$.
\end{lm}

 Let us define the main decision problem of this section, which generalises zeroness of polynomial grammars without substitution.

\noindent\textbf{\Problemtwo}
\\
\textbf{Input:}\\
$R$ - \emph{a computable ring with no zero divisors},\\
$X$ - \emph{finite set of variables},\\
$A$ - \emph{\pg over ring $R[X]$}, \\
$B$ - \emph{\pg over ring $R$}
\\
\textbf{Question:} 
\emph{Is it the case that for all $a \in A, b \in B$}
$
a(b) = 0?
$
\\
We often abbreviate the statement above to
$$
A(B) = 0.
$$
This problem can be seen as a special case of zeroness of \modeltwoundecidables, for grammars of the form $\subs(A,B)$, where $A,B$ are polynomial grammars. We call such substitution \emph{independent} because tuples of substituted values come from a different nonterminal than polynomials to which they are substituted, and therefore they are produced independently; we prove this problem decidable in Theorem \ref{thm:modeltwo-decidable}. As we will see in Section \ref{subsec:modeltwoundecidable-undecidability}, in case of \emph{dependent} substitution, the problem is undecidable.
\begin{example}
	Consider a polynomial grammar $A$ over ring of polynomials in two variables $\wt x, \w x$ with coefficients in $\Q(\wt a, \w a)$, i.e. $(\Q(\wt a, \w a))[\wt x, \w x]$ defined as $A \ra m(S), S \ra p(S)\mid (0,1,0,1)$ where $m(\wt r, \w r, \wt s, \w s) = \wt r - \wt s$ and $p(\wt r, \w r, \wt s, \w s) = (\wt r \cdot \w x \cdot \w a + \wt x \cdot \w a + \wt a, \w r \cdot \w x \cdot \w a, \wt x \cdot \w s \cdot \w a + \wt s \cdot \w a + \wt a, \w x \cdot \w s \cdot \w a)$
	and a polynomial grammar $B$ over ring of polynomials in two variables $\wt a, \w a$ and coefficients in $\Q$, i.e. $\Q[\wt a, \w a]$ defined as $B \ra q(B) \mid (0, 1)$ where $q(\wt t, \w t) = (\wt t \cdot \w a + \wt a, \wt t \cdot \w a$).
	Then $A(B)=0$.
\end{example}
\begin{remark}
	The above Example expresses the fact, under our string-to-polynomial encoding, that $(xa)^n = x^na^n$ for all words $x \in \{a\}^*$, for all $n\geq 0$.
\end{remark}

%
\partofsec{Remark on algebraic closure.} Let us discuss a subtlety before we state and prove Theorem \ref{thm:modeltwo-decidable}. Despite the fact that all polynomials arising from our string-to-polynomial encoding are elements of field \QSigF, we carry out the proof for algebraically closed fields.
It does not affect generality, as every field can be effectively embedded into algebraically closed field (Lemma \ref{fact:rabin-alg-clos-computable}).
We do so because it makes the description of the algorithm more readable, as in such case one can compute the coordinate ring of a given algebraic set (Lemma \ref{lm:coordinate-ring-computable-from-V}); however, we present a version of the proof that does not use the notion of algebraic closure in Section \ref{subsec:no-alg-closed}.

We conjecture that the algorithm from Theorem \ref{thm:modeltwo-decidable} does \emph{not} have to use elements of \barK that are not in $K$ -- in consequence, the decision of introducing algebraic closure would affect only the high-level description of the algorithm, not the algorithm itself. This might be important when trying to obtain a feasible complexity for a special case (recall in general zeroness of polynomial grammars over a ring is Ackermann-hard \cite[Theorem 1]{WorrellEtAl17}).

\subsection{\Problemtwo: decidability}
In this section we prove Theorem \ref{thm:modeltwo-decidable}.
\begin{theorem}\label{thm:modeltwo-decidable} 
Let $R$ be a computable ring with no zero divisors. Then \problemtwo over $R[X]$ is decidable.
\end{theorem}
As stated in Remark \ref{rm:reduction-to-an-undecidable-problem}, the reduction from Lemma \ref{lm:RTsubss-can-be-simulated-by-pgsubss}, restricted to the special case from Theorem \ref{thm:modeltwo-decidable}, yields decidability of the corresponding problem for register transducers (Section \ref{subsec:interpretation-as-sys-of-eqs-testing}).

The proof is based on the proof of \cite[Corollary 8.2]{SeidlEtAl2018} in the presentation of \cite{BojSiglog2019}.
\begin{proof}[Proof of Theorem \ref{thm:modeltwo-decidable}]
	Without loss of generality we may assume that grammar $A$ is of dimension 1 (notice we cannot assume it for $B$). In this proof, by abuse of notation, we identify a \pg $A$ with its the language. We treat $A$ and $B$ as polynomial grammars over a computable algebraically closed field that contains $R$, for example algebraic closure of its field of fractions.
	
	There is a clear semi-procedure for non-zeroness -- it essentially amounts to enumerating all derivations. In rest of the proof, we show a semi-procedure for zeroness.
	
	Observe that $A(B) = 0$ if and only if for some algebraic set $V$:
	$$\begin{cases}
				B \subseteq V \text{ (i) and}  \\
				 A\restr{V} = 0 \text{ (ii)},
	\end{cases}$$
where $A\restr{V}$ denotes set $\{f\restr{V} \mid f \in A\}$.
Indeed -- for example, as $V$, one may take $V(A)$.
	The semi-procedure for equivalence is as follows: guess, by infinite enumeration (see Corollary \ref{cor:algebraic-sets-enumerate}), an algebraic set $V$ and test conditions (i) and (ii), which can be done the following way.
	\begin{sublm}
		For a given $V$, condition (i) can be tested effectively.	
\end{sublm}
	\begin{proof}
		Condition (i) reduces to zeroness of \pgs in the following way. Denote by $I$ the ideal that $V$ is represented by. Observe that $B \subset V(I) = V$ if and only if $f(B) = 0$ for all generators $f$ of $I$.
	\end{proof}
	\begin{sublm}\label{sublm:condition-2}
		For a given $V$, condition (ii) can be tested effectively.
	\end{sublm}
	\begin{proof}[Proof sketch]
		Decompose $V$ into (irreducible) varieties $V = V_1 \cup V_2 \cup \ldots \cup V_k, \ \! \! k \geq 1$ (see Lemma \ref{lemma:deco-into-irreducibles}). Then $A\restr{V} = 0$ holds iff $A\restr{V_i} = 0 \text{ for } i =1,\ldots, k$. Observe that $A\restr{V_i} = 0$ is equivalent to zeroness of $A$, treated as a \pg over coordinate ring of $V_i$ (see Lemma \ref{lm:coordinate-ring-computable-from-V} for effectiveness); it has no zero divisors (Lemma \ref{lm:irreducible-iff-coordinate-ring-no-zero-divisors}) and hence decidability follows from Lemma \ref{lm:pg-over-K}.
	\end{proof}
\end{proof}

\partofsec{Generalisation.}
	An analogous problem can be defined for arbitrary number of grammars; for $k$ grammars, we will abbreviated it to $A_1(A_2(\ldots (A_{n-1}(A_k)))) = 0$. Let us sketch the proof of its decidability -- it is very similar to the proof of case $n = 2$. 
	Observe that $A_1(A_2(A_3(\ldots (A_n)\ldots))) = 0$ is equivalent to $A_2(A_3(\ldots (A_n)\ldots)) \subseteq V$ and ${A_1}\restr{V} = 0$ for some algebraic set $V$. 
	For a given $V$, first condition is equivalent to $f(A_2(A_3(\ldots (A_n)\ldots))) = 0$ for generators $f$ of ideal that $V$ is represented by -- this is decidable by induction assumption for $n-1$ polynomial grammars $f( A_2), A_3, \ldots, A_n$. Second condition was proven decidable in Sublemma \ref{sublm:condition-2}. 
\begin{remark}
	For any \pg $A$ there always exists a \emph{finite} set of polynomials $F$ such that $A(v) = 0$ if and only if $F(v) = 0$, for every vector $v$. In particular $A(B)=0$ if and only if $F(B) = 0$. Given a finite set of polynomials $F$, $F(B)=0$ can be decided (Lemma \ref{lm:pg-over-K}). However, it is not clear how to compute such set $F$  given grammar $A$.
\end{remark}
\begin{remark}
	Let $A,B$ be polynomial grammars. If the grammar $A$ is linear, i.e. there is at most one nonterminal on the right-hand side of every production rule, then grammar $A(B)$ can be seen as a polynomial grammar too -- it is enough to ``concatenate'' those two grammars by replacing production rules of grammar $A$ of form $nonterminal \ra constant$ with $nonterminal \ra B$. However in case $A$ is not linear, it is not clear how to convert $A(B)$ to a polynomial grammar.
\end{remark}
\subsection{\Modeltwoundecidable: undecidability}\label{subsec:modeltwoundecidable-undecidability}
In previous section we saw that restriction to one \emph{independent} substitution per derivation yields decidability of zeroness. In this section we show that in general case of one \emph{dependent} substitution, this problem is undecidable. We call such substitution dependent because substituted tuples and polynomials to which they are substituted may come from the same nonterminal and hence be generated dependently.

\begin{restatable}{thmrestat}{thmModeltwoundecidableUndecidable}\label{thm:modeltwoundecidable-undecidable}
	Zeroness of \modeltwoundecidables over $\Z[x]$ is undecidable.
\end{restatable}
In contrast to Theorem \ref{thm:modeltwo-decidable}, the reduction from Lemma \ref{lm:RTsubss-can-be-simulated-by-pgsubss} restricted to special case from Theorem \ref{thm:modeltwoundecidable-undecidable} does \emph{not} yield decidability nor undecidability of analogous problem for register transducers.

The proof was suggested by Lasota and Pi{\'o}rkowski \cite{SlawekRadekPersonal2019}. It is analogous to proof of  \cite[Theorem 17]{BoiretEtAl2018}, except two-counter machines are replaced by reset VASS-es; used encoding of reset VASS into register transducer is similar to \cite[Example 3]{BenediktDuffSharadWorrell2017}.

	By a \emph{VASS} we mean a Vector Addition System with States. We call vectors that describe transitions \emph{step vectors}.
	By a \emph{reset VASS} we mean an extension of VASS where transitions may also reset some of coordinates to 0.
	By a \emph{unit step vector} we mean a step vector of form $(0,0,\ldots,0,1,0,\ldots,0)$.
	For a reset VASS, by \zerovectVASS we denote vector of zeros of its dimension.

\begin{proof}[Proof sketch]
	We show a reduction from reachability in reset VASS, which is known to be undecidable \cite[Theorem 5]{ArakiKasami76}. 
	Without loss of generality assume it is reachability from \zerovectVASS  to \zerovectVASS  and step vectors are unit -- this ensures that sum of coordinates is an integer from interval $[0,n]$ after $n$ steps of a valid run. Given a reset VASS $\V$ we construct a register transducer with substitution \emph{over ring $\Z[x]$} that uses substitution only in the output transition (hence once per run) such that $\V$ can reach \zerovectVASS from \zerovectVASS if and only if the register transducer returns a non-zero output on some input (such register transducers are a different formalism for the same model as \modeltwoundecidables).
	
	We describe the construction. The transducer takes a sequence of transitions of $\V$ as an input (this is a word over finite alphabet), which may describing a (valid) run of \V or not. The state of the transducer is the state of $\V$ and turns to error state if in the input word the states of some pair of consecutive transitions do not match. The transducer holds the counter values of $\V$ in registers that we call \emph{counter registers}. It has an \emph{error register} which turns 0 when an ``error'' occurs, i.e. some coordinate of $\V$ goes below 0 (note that step vectors are unit so such coordinate equals $-1$) and holds 0 for the rest of the run -- this is achieved by multiplying the error register by the product of values of counter registers, each value incremented by 1, in each update. There is a \emph{reachability test register} that holds a polynomial $(x-1)(x-2)\ldots (x-n)$ after $n$ steps (formally, to construct it one needs an \emph{auxiliary register} that holds number $n$); it has a property that, for an integer input from interval $[0,n]$, it evaluates to 0 if and only if the input number is non-zero. Finally, if the state is accepting, in output transition the transducer returns error register multiplied by reachability test register evaluated at the sum of counter registers (this involves substitution of a register into another register); observe this is a non-zero number if and only if input word is a (valid) run of \V and this run reaches \zerovectVASS from \zerovectVASS.
	We give details of the construction in Section \ref{subsec:undecidable-symbols}.
\end{proof}
\subsection{Interpretation of Theorem \ref{thm:modeltwo-decidable} and Theorem \ref{thm:modeltwoundecidable-undecidable} as testing infinitely many equations on a language generated by a register transducer} \label{subsec:interpretation-as-sys-of-eqs-testing}
Consider a register transducer $E$ that returns pairs of words over alphabet $\Delta \cup X$ and a register transducer $T$ that returns $X$-tuples of words over $\Delta$. Transducer $E$ can be interpreted as a generator of equations (constraints) and $T$ as generator of tuples to be tested. Then Theorem \ref{thm:modeltwo-decidable} states that it is decidable if all tuples of words generated by $T$ satisfy all equations generated by $E$. On the other hand, Theorem \ref{thm:modeltwoundecidable-undecidable} states that, in case when polynomials and integers are generated, it is undecidable for a \pg that generates ``pairs'' (equation, tuple of vectors) if the ``second'' coordinate is always a solution for the ``first'' one (formally, the numbers of mentioned coordinates differ as e.g. an equation alone requires two coordinates). 
In such interpretation, words ``independent'' and ``dependent'' describe the relationship between equations and tested tuples.
We include a formal presentation in Section \ref{subsec:formal-pres-of-interpretation}.

%% file: modelone.tex
In previous section we considered \pgs with restricted number of uses of substitution per derivation. Another way of restricting substitution is to consider evaluations -- instead of adding general ($1+|X|$) - ary substitution operation 
$$
\subs:K(X) \times K(X)^{X} \ra K(X)
$$
 we add \emph{evaluations}, which are substitutions where the argument from $K(X)^{X}$ is fixed, i.e. unary operations 
 \begin{align*}
 &\eval_{v}: K(X) \ra K(X),	\\
&\eval_{v}(p) \mapsto p(v),
 \end{align*}
for vectors  $v \in K(X)^{X}$.
We call this restriction of \pgsubss \emph{\pgevals}. 

In this section we prove Theorem \ref{thm:modelone-decidable} that states that \pgevals that satisfy \simautcondition (Definition \ref{df:pg-aut}) have decidable zeroness. 
The core technical ingredient is Lemma \ref{lm:pgaut-decidable}, which is a result of discussions with Boja\'{n}czyk, Worrell, Shirmohammadi and Kiefer \cite{WorrellPersonal}, allows to use simultaneous field automorphisms in \pgs over a field. Lemma \ref{lm:com-inj-induce-automorphisms} gives an easy-to-check characterisation of word substitutions that 
are mapped to field automorphisms by the string-to-polynomial encoding from Definition \ref{df:str-to-pol-encoding}.
Theorem \ref{thm:modelone-decidable} concludes those two lemmas.
\subsection{\Modelone: decidable equivalence.}
\begin{df}
	Let $R$ be a ring. An \emph{automorphism} of $R$ is a bijective function $h:R \ra R$ that satisfies $h(\op(a, b)) = \op(h(a), h(b))$ for $\op$ being either addition or multiplication for every $a, b \in R$.
	For example, evaluation $\eval_v$ for $v = x+1$, i.e. 
	$$
	\Q(x) \ni f \mapsto f[x:=x+1] \in \Q(x)
	$$
	 is an automorphism of the field $\Q(x)$.
	An automorphism of  product of fields $K^n$ is called \emph{simultaneous} if it is of form $\alpha^n = \underbrace{(\alpha, \alpha, \ldots, \alpha)}_n$ for some automorphism $\alpha$ of $K$.
\end{df}
\begin{df}\label{df:pg-aut}
	A \emph{\pgaut/evaluations} is an extension of \pg over a field where production rules are of form $Y \ra p(\alpha^{N}(Y_1, \ldots Y_n))$ where $\alpha$ is a field automorphism/an evaluation and $N=\dim Y_1 + \ldots + \dim Y_n$.
\end{df}
\begin{lm}\label{lm:pgaut-decidable}
	Let $K$ be a computable field.
	Then \pgauts over $K$ have decidable equivalence.
\end{lm}
Similarly as for Theorem \ref{thm:modeltwo-decidable}, the reduction from Lemma \ref{lm:RTsubss-can-be-simulated-by-pgsubss}, restricted to the special case from Lemma \ref{lm:pgaut-decidable}, yields decidability of equivalence for certain class of register transducers with substitution (Theorem \ref{thm:modelone-decidable}).
\begin{proof}[Proof of Lemma \ref{lm:pgaut-decidable}]
	A proof analogous to proof of Lemma \ref{lm:pg-over-K} can be performed for {\pgauts} -- it is enough to show that preimage of an algebraic set by simultaneous automorphism of  $K$ is effectively an algebraic set (note that this fact holds for any algebra, with the same proof as for fields).
	
	Let $\alpha$ be an automorphism of $K$; by abuse of notation, denote the same way coordinate-wise application of it.
	Then  $\alpha^{-1}(V(\enum p1n)) = V(p_1\circ \alpha, \ldots, p_n \circ \alpha)$. Functions $p_i \circ \alpha$ might not be polynomial, however for a polynomial $p(\enum x1n) = \sum_{i=0}^n a_ix^i$ equality $(p\circ\alpha)(\enum x1n) = 0$ is equivalent to a polynomial equation  $\sum_{i=0}^n \alpha^{-1}(a_i)x^i = 0$. 
\end{proof}

\begin{restatable}{defrestat}{dfComInjective}\label{df:com-injective}
	We call a word substitution $p : \Sigma \ra \Sigma^*$ \emph{com-injective} if $w \mapsto p(w)$ is injective, when treated as a mapping of commutative words.
\end{restatable}
By abuse of notation, we also denote as com-injective a polynomial substitution, which is an image of com-injective word substitution by our string-to-polynomial encoding.

\begin{lm}\label{lm:non-vanishing-is-cominjective}
	If $p$ is a substitution for one letter, call it $\sigma$, (in other words, $p$ is an identity function on all letters except $\sigma$) and it does not vanish letter $\sigma$, then it is com-injective. Example of such substitution is $p: a \mapsto bacbaab$. A non-example of such substitution is $p: a \mapsto bbc$.
\end{lm}
\begin{proof}
	Straightforward.
\end{proof}
In the following lemma we characterise word substitutions that induce field automorphisms via our string-to-polynomial encoding. Finding such word substitutions is challenging, as almost no word substitution induces an automorphism of field $\QSigF$. However, a reasonably large class of word substitutions induces automorphisms of the field $\QQzero$, an extension of \QSigF that is defined below.

Notions of a polynomial and a rational function, considered as formal expressions, can be extended to situation where some of variables have exponents from an arbitrary infinite monoid, that can be embedded in some group -- ``ordinary'' polynomials use monoid $(\mathbb N, + ,0)$. We use this extension with monoid $(\mathbb Q_{\geq 0}, +, 0)$. Let $K$ be a field. By \KQzeroR we denote the ring of polynomials over variables $\wt \Sigma \cup \w \Sigma$ where variables $\w \Sigma$ have exponents in $(\mathbb Q_{\geq 0}, +, 0)$ and by $\KQzeroF = \{\frac{p}{q} \ | \ p, q \in \KQzeroR, q\neq 0\}$ the corresponding field of rational functions; let us emphasize that the latter contains elements with negative exponents as well, e.g. ${\w a} ^{-1/2} = \frac{1}{{\w a} ^{1/2}}$.

\partofsec{A subtlety.} There is a subtlety regarding defining evaluations. In field \KQzeroF, as opposed to \KSigF, evaluation is \emph{not} defined for every vector: for it to do so, it must evaluate variables from $\w\Sigma$ on monomials in variables from $\w\Sigma$; however polynomial substitutions induced by word substitutions via our word-to-polynomial encoding have that property.

\begin{restatable}{lemmarestat}{lmComInjInduceAutomorphisms}\label{lm:com-inj-induce-automorphisms}
	Com-injective word substitutions induce automorphisms of field $\QQzero$.
	\\
	Non-com-injective word substitutions do not induce automorphisms of any field that contains \QSigF.
\end{restatable}
\begin{cor}
	Grammar $Y$ from Introduction is a \pgaut via our string-to-polynomial encoding.
\end{cor}
\begin{proof}
	Substitutions in the grammar are simultaneous -- both register transducers apply the same word substitutions at the same steps on all of their registers 
	(notice that update of register $S$ does not use substitution explicitly, but it could with no difference, as strings stored in $S$ do not contain symbol $\#$, for which the substitution is performed). They are com-injective \ref{lm:non-vanishing-is-cominjective} and hence induce automorphisms (Lemma \ref{lm:com-inj-induce-automorphisms}).
\end{proof}
From Lemma \ref{lm:com-inj-induce-automorphisms} one can conclude Theorem \ref{thm:modelone-decidable}.
\begin{theorem}\label{thm:modelone-decidable}
\ModeloneRTs have decidable equivalence.
\end{theorem}
\begin{proof}
	Due to Lemma \ref{lm:com-inj-induce-automorphisms}, the corresponding \pgeval is a \modelone, when considered over the field \QQzero. Its zeroness is decidable due to Lemma \ref{lm:pgaut-decidable}.
\end{proof}
Now we give a proof of Lemma \ref{lm:com-inj-induce-automorphisms} for an example substitution, in order to to avoid unnecessary formalism; the proof idea is the same in the general case, for which we give a proof in Section \ref{subsec:proof-thm:modelone-decidable}.
\begin{proof}[Proof of Lemma \ref{lm:com-inj-induce-automorphisms} for an example substitution]
	Let $\Sigma = \{a,b\}$. Consider a substitution of words $p: a^* \ra a^*$ defined as $p: a \mapsto aa$. Then induced substitution of polynomials is $p': \widetilde{a} \mapsto \widetilde{a}\overline{a} + \widetilde{a}, 
	\overline{a} \mapsto \overline{a}^2.$
	This substitution is not invertible when considered in the field \QSigF, however we show it is invertible when considered in the field \QQzero. To find the inverse substitution, we consider a system of equations:
	\begin{align}\label{eq:sys-of-eq}
		\begin{cases}
		\widetilde{a}' = \widetilde{a} \cdot (\overline{a} + 1),   \\
		\overline{a}' = \overline{a}^2.
		\end{cases}
	\end{align}Second equation is purely in variable $\w a$, hence we solve it first; it is a linear equation, where the exponents play the role of coefficients from the field $\Q$. We get $\w{a} = {\w{a}'}^{1/2}.$ We substitute this into the first equation and get 
	$
	\widetilde{a}' = \widetilde{a} \cdot ({\overline{a}'}^{1/2} + 1), 
	$
This is a linear equation in variable $\wt{a}$ with coefficients in the field $\Q(\wt a', {\w{a}'}^{\Q_{\geq 0}})$. Solving it, we obtain
	$
		\widetilde{a} = \frac{\widetilde{a}'}{1 + {\overline{a}'}^{1/2}}.
		$
	This shows that substitution $\widetilde{a} \mapsto \frac{\widetilde{a'}}{1 + {\overline{a}'}^{1/2}}, 
	\overline{a} \mapsto {\overline{a'}}^{1/2}$
	is the right-inverse of $p$. It is a two-sided inverse because the equations are linear (in the mentioned sense). This finishes the proof for this example.
\end{proof}

\begin{remark} Let us notice that for some grammars, Theorem \ref{thm:modelone-decidable} can be reduced to zeroness of \pgs \emph{without} substitution by simply removing the automorphisms from production rules -- 
	for example this is the case when polynomial functions used in productions have integer coefficients.
However this is not true in general when ``new'' occurrences of substitutable constants are introduced in the productions.
For example, consider \pgsubs over $\Q[a,b]$ with production rules 
$$
S \ra p(X), X \ra q(Y), Y \ra a-b,
$$ where $p,q:\Q[a,b] \ra \Q[a,b]$ are defined as $p(f) = f+b, q(f) = f[b:=a+b]$. Then $S=0$, but if automorphisms were removed, we would have $S = a$.
\end{remark}

%% file: summary.tex
Theorem \ref{thm:modeltwoundecidable-undecidable} draws a rather pessimistic view on \pgsubss over ring of polynomials -- even almost the simplest model has undecidable zeroness. In this paper we show positive results when the model is restricted either to independent substitution (Theorem \ref{thm:modeltwo-decidable}) or to evaluations (Theorem \ref{thm:modelone-decidable}). Theorem \ref{thm:modeltwo-decidable} can be interpreted as decidability of testing infinite systems of equations on languages, where both the system and the language can be generated by register transducers (Section \ref{subsec:interpretation-as-sys-of-eqs-testing}).
Theorem \ref{thm:modelone-decidable} seems to have large limitations: first, every evaluation must be ``simultaneous'' i.e. in each production the same evaluation has to be applied to all coordinates of all nonterminals at once, second, this evaluations needs to induce a field automorphism, which excludes evaluations that vanish letters as they are not an injective.
 Having said that, in Introduction we define grammar $Y$ that satisfies both these conditions and can be used to decide equivalence of two register transducers (also defined in Introduction) that compute function $\sqrev$ .

Finally, we ask open questions. First one is about expressiveness of register transducers with evaluations (no automorphism condition).
\begin{enumerate}
	\item What are other interesting string-to-string functions computed by register transducers with evaluations whose evaluations are not allowed to vanish letters?
	\setcounter{enumi_val}{\value{enumi}}
\end{enumerate}
Function similar to $\sqrev$, that maps each word $w$ to $\rev(w)^{|w|}$ seems not to be an example, as it seems to require introducing symbol $\#$ in a way analogous to $\sqrev$ and then vanishing it.

Also, it is interesting what problems about register transducers are decidable, in particular:
\begin{enumerate}
	\setcounter{enumi}{\value{enumi_val}}
	\item Is equivalence undecidable?
	\item Is equivalence decidable if we restrict to register transducers with evaluations (this question could be asked also for polynomial grammars)?
\end{enumerate}

%% file: appendix.tex
\subsection{Formal presentation of Section \ref{subsec:interpretation-as-sys-of-eqs-testing}}\label{subsec:formal-pres-of-interpretation}
In this section we give a formal presentation of the problem presented in Section \ref{subsec:interpretation-as-sys-of-eqs-testing}.
For a \pg $A$, by abuse of notation we identify it with its language. We represent equations as a pair of polynomials, denoted as $e = (e_1, e_2)$; a potential solution $v$ is a vector of polynomials; $v$ is a solution of $e$ iff $e_1(v) = e_2(v)$.
\\
\\
\textbf{Independent equation satisfiability}
\\
\textbf{Input:}\\
$\Delta$ - \emph{finite alphabet}, \\
$X$ - \emph{finite set of variables},
\\
$E\subseteq ((\Q[\Delta])[X])^2$ - \emph{set of equations in variables $X$}, given by \pg $E$ of dimension 2.
\\
$T\subseteq (\Q[\Delta])^{|X|}$ - \emph{tested language}, given by \pg $T$ of dimension $|X|$,
\\
\textbf{Question:} 
Is it the case that for all inputs $(e_1, e_2) \in E, v\in T$: 
$$
	e_1(v) = e_2(v)?
$$
\\
\textbf{Dependent equation satisfiability}\\
\textbf{Input:}\\
$\Delta$ - \emph{finite alphabet}, \\
$X$ - \emph{finite set of variables},
\\
$G$ - polynomial grammar over $(\Q[\Delta])[X]$
of dimension $(|X|+ 2)$
\\
\textbf{Question:} 
Is it the case that for all inputs $((e_1, e_2), v) \in G$: 
$$
e_1(v) = e_2(v)?
$$
\subsection{Detailed proof of Theorem \ref{thm:modeltwoundecidable-undecidable}}
\label{subsec:undecidable-symbols}
In this section we give a details of the construction of register transducer from proof of Theorem \ref{thm:modeltwoundecidable-undecidable}.
\thmModeltwoundecidableUndecidable*
	\begin{proof}[Details of construction from \ref{thm:modeltwoundecidable-undecidable}]	
	Given a reset VASS $\V$ we construct a register transducer with substitution over ring $\Z[x]$ that uses substitution once per run such that $\V$ can reach \zerovectVASS from \zerovectVASS if and only if the register transducer returns a non-zero output on some input. We omit description of the states as it is clear.
	 
	\noindent \textbf{Description} of register transducer: \\
	\emph{Registers:}\\
	\indent $R_1$: will store $\mathbb{Z}[x]$ (\emph{reachability test register}), \\
	\indent $R_1^{(aux)}$: will store $\mathbb{Z}$ (\emph{auxiliary register}),\\
	\indent $R_2$: will store $\mathbb{Z}$ (\emph{error register}),\\
	\indent $S_1, S_2, \ldots, S_{dim}$: will store $\mathbb{Z}$ \emph{(counter registers)} .\\
	\emph{Input alphabet:}\\
	\indent $\Sigma = $ set of transitions of $\V$.\\
	\emph{Output function:}\\
		\indent
		$
		\out = 
		\begin{cases}
			0, &\text{if state is rejecting}, \\
			R_1[x:=S_1 + \ldots + S_{dim}]\cdot R_2, &\text{else}.
		\end{cases}
			$
			
		\textbf{Invariants.} Observe, that run of reset VASS can be identified with a word over an alphabet consisting of its transitions. Registers after reading a run of $\V$ of length $n$ satisfy the following: 
		\begin{enumerate}
			\item $S_1, \ldots, S_{dim} \in \mathbb{Z}$ are current coordinates of $\V$, viewed as $\mathbb{Z}$ - VASS (i.e. can go below 0),
			\item $R_2 = 0$ iff some coordinate of \V went below 0 (hence $R_2$ checks correctness of the run),
			\item $R_1[x:=S_1 + \ldots + S_{dim}] \neq 0$ iff current point is \zerovectVASS,
			\begin{enumerate}
			\item $R_1[x:=i] = 0$ iff $i!=0$ for integer $i$ from range $[0, n]$,
			\item $R_1^{(aux)} = n$.
			\end{enumerate}
			\end{enumerate}
			
			\noindent \textbf{Construction} of \modeltwoundecidable:
			\\
			\textbf{Initial values} of registers:
			\begin{enumerate}
			\item $S_1, \ldots, S_{dim} = 0$
			\item $R_2 =1$:
			\item $R_1$ :
			\begin{enumerate}
			\item $R_1 = 1$,
			\item $R_1^{(aux)} =0$.
			\end{enumerate}
			\end{enumerate}
			\textbf{Update of registers}:
			\begin{enumerate}
			\item $S_1, \ldots, S_{dim}$ are incremented, decremented or reset to 0, according to read transition,
			\item $R_2$:
			\begin{enumerate}
			\item $R_2 = R_2 \cdot (S_1'+1) \cdot \ldots \cdot (S_{dim}'+1) $, where $S_i'$ is the value of $S_i$ after the update.
			\end{enumerate}
			\item $R_1$ :
			\begin{enumerate}
			\item $R_1 = R_1 \cdot (x - R_1^{(aux)})$ (recall $x$ is a constant in $\Z[x]$),
			\item $R_1^{(aux)} = R_1^{(aux)} +1$.
			\end{enumerate}
			\end{enumerate}
			\end{proof}
\subsection{Proof of Theorem \ref{thm:modeltwo-decidable} without algebraic closure}\label{subsec:no-alg-closed}
In this section we present the proof of Theorem \ref{thm:modeltwo-decidable} without introducing the notion of algebraic closure of a field.
\begin{obs}[Representing algebraic sets -- not unique]\label{no-alg-clos:obs:representation-algebraic-sets-not-unique}
There is a subtlety regarding representing algebraic sets. Algebraic sets are \emph{not} represented by ideals of polynomials that zero on them ($I(V)$) but by ideals that they are zeros of ($I \text{ such that } V = V(I)$) -- these two notions coincide in case field is algebraically closed and radical of $I$ is taken (taking radical of $I$ does not change $V(I)$) due to Fact \ref{fact:hilbert-nullstellensatz-cor}. The former is a unique representation, but we do not choose it because of computability issues. It is not clear how to decide for a given ideal, even for a radical one, if it is of form $I(V)$ (it might be merely contained in $I(V)$ for $V=V(I)$). Therefore, it is not clear how to enumerate such ideals and, if we chose it as the representation , algebraic sets.

This subtlety is relevant to Sublemma \ref{sublm:condition-2}. Assume that $V$ was given by an ideal that is smaller than $I(V)$. Then there are two options:
\begin{itemize}
\item test for condition (2) is passed, which gives a true positive (we show that in a moment), or
\item test for condition (2) is failed, potentially giving false negative. Such event fortunately will be covered either by a true positive when enumerating $V$ (again) by $I(V)$ or by some previously enumerated true positive.
\end{itemize}
\end{obs}
\begin{obs}\label{no-alg-clos:obs:coordinate-ring-computable}
Let algebraic set $V$ be given by ideal $I$, i.e. $V = V(I)$. Then coordinate ring of $V$ is a quotient of $K[X]/I$.
\end{obs}
	\begin{proof}
	Straightforward from $I \subseteq I(V)$, which follows from definition of $I(V)$.
	\end{proof}
	Before we prove the main result of this section, let us state the following lemma.
	\begin{lemma}(\cite[Chapter 4, §8, Corollary 10 and remarks at the end of Chapter 4, §8]{CoxOShea2015})\label{lemma:prime-deco}
	Let $K$ be a computable field. Then every radical ideal $I$ effectively admits a \emph{prime decomposition}.
	\end{lemma}
Now we are ready to present the main proof.
\begin{proof}[Proof of first bullet point]
Let us go through the algorithm for Sublemma \ref{sublm:condition-2} and see if indeed positives are true, even for $V$ being represented by radical ideal $I$ contained in $I(V)$.
If $A$ is zero, treated as a \pg over $K[X]/I$, the more it is when zero treated as a \pg over coordinate ring (Observation \ref{no-alg-clos:obs:coordinate-ring-computable}). Is crucial that when decomposing $V$ into irreducibles, we use algorithm for prime decomposition (Lemma \ref{lemma:prime-deco}) of the ideal -- then this ring has no zero divisors (Lemma \ref{lm:irreducible-iff-coordinate-ring-no-zero-divisors}) and hence decidability follows from Lemma \ref{lm:pg-over-K}.
\end{proof}
\subsection{Proof of Lemma \ref{lm:com-inj-induce-automorphisms}} \label{subsec:proof-thm:modelone-decidable}
In this section we give the proof of Lemma \ref{lm:com-inj-induce-automorphisms} used in the proof of Theorem \ref{thm:modelone-decidable}.
\lmComInjInduceAutomorphisms*
\begin{lemma}\label{lm:com-inj-(i)-and-(ii)}
	Let $p : \Sigma \ra \Sigma^*$ be a com-injective word substitution. The following two linear mappings of $|\Sigma|$-dimensional linear spaces are invertible:
	\begin{enumerate}[(i)]
		\item $\wt w \mapsto \wt{p(w)}$ of linear space of formal linear combinations of variables $\wt\Sigma$ with coefficients in the field $\Q(\w\Sigma)$,
		\item $w \mapsto p(w)$ treated as a mapping of commutative words, extended to \Q-linear space $\Q^{\Sigma}$
	\end{enumerate}
\end{lemma}
We illustrate the statement of above lemma with the following example.
\begin{ex}
	Let $\Sigma = \{a,b\}$ and $p: a \mapsto ab, b \mapsto babb$.
	\\
	Values of $\QSigF$-linear mapping $p$ in basis $\wt\Sigma$ are:
	\begin{align*}
		p(\wt a) &= \wt{ab} = \wt a \w b + \wt b, \\
		p(\wt b) &= \wt{babb} = \wt b \w {a}{\w b}^2 + \wt a {\w b}^2 + \wt b \w b + \wt b .
	\end{align*}
	Values of $\Q$-linear mapping $p$ on basis $\Sigma$ (coefficients are in exponents) are:
	\begin{align*}
		p( a) &= ab = a^1 b^1, \\
		p(b) &= babb = a^1b^3,
	\end{align*}
	(recall in (ii) outputted words are considered commutative) hence the corresponding matrices are 
	$
	\begin{bmatrix}
		\w b & {\w b}^2 \\
		1 & \w {a} {\w b}^2 + \w b + 1
	\end{bmatrix}
	$ and 
	$
	\begin{bmatrix}
		1 & 1 \\
		1 & 3
	\end{bmatrix}.
	$
	They are invertible: determinants are equal to, respectively, $\w a {\w b}^3 + \w b$ and 2. 
\end{ex}
\begin{proof}[Proof of Lemma \ref{lm:com-inj-(i)-and-(ii)}]
	(ii) is straightforward from definition of com-injectivity. Observe that mapping from (ii) is the same as mapping from (i) when 1 is substituted for variables from $\w\Sigma$; the same holds for their determinants, and hence if determinant for (ii) is a non-zero number, determinant for (i) is a non-zero polynomial.
\end{proof}
No we proceed towards the proof of Lemma \ref{lm:com-inj-induce-automorphisms}.
By \emph{explicit form} of a homomorphism we mean a substitution that defines it; a homomorphism given in this form is said to be \emph{given explicitly}.
We write systems of equations as (name of equation: equation); we define subsystems by giving equation names.
\begin{proof} [Proof of Lemma \ref{lm:com-inj-induce-automorphisms}]
Assume substitution $p$ is com-injective.
	To find the inverse substitution (if exists) we consider system of equations $S = \{\wt \sigma: \wt{\sigma}' = \wt{p(\sigma)}\} \cup \{\w \sigma: \w{\sigma}' = \w{p(\sigma)}\}$.
\w\Sigma-subsystem of $S$ is purely in variables \w\Sigma hence we solve it first. Right hand sides are monomials and it can be seen as system of linear equations, where exponents play the role of coefficients from field \Q; we solve it (solution exists -- its existence is equivalent to (ii) of Lemma \ref{lm:com-inj-(i)-and-(ii)}). This solution gives explicitly an automorphism of $K(\w a^{\Q}, a\in \Sigma)$ which naturally is also an automorphism of \KQzeroF.
We substitute this to \wt\Sigma-subsystem and obtain a system of linear equations in variables \wt\Sigma' and coefficients in \KQzeroF; this system is automorphic to \wt\Sigma-subsystem of $S$ and hence has a solution (due to (i) of Lemma \ref{lm:com-inj-(i)-and-(ii)}).
Its solution, together with solution of \w\Sigma-subsystem, gives the right-inverse substitution explicitly. It is a two-sided inverse because both systems of equations were linear (in the mentioned sense).

The proof of the converse is analogous, which finishes the proof.
	\end{proof}
	
\subsection{Proof of Lemma \ref{lm:coordinate-ring-computable-from-V}}\label{subsec:computation-rings-and-ideals}
In this section we prove Lemma \ref{lm:coordinate-ring-computable-from-V} used in the proof of Theorem \ref{thm:modeltwo-decidable}. It is a well known result (e.g. implicit in \cite[Chapter 5, §2 and §3 and §4]{CoxOShea2015}), but we did not find an explicit reference.

\partofsec{Preliminaries.}
By $K[X]$ we denote over some field $K$ and a finite set of variables $X$.

Every ideal $I$ in ring $K[X]$ induces a congruence $\simI$ defied as $f \simI g$ iff $f-g \in I$, for $f, g \in K[X]$. By $K[X]/I$ we denote \emph{quotient ring} $K[X]/\!\!\simI$. For $f \in K[X]$.
A \emph{radical of an ideal} $I$ is the set $\sqrt{I} = \{f \in K[X] \mid f^n \in I \text{ for some } n \geq 0\}$; it is an ideal. An ideal $I$ is \emph{radical} if $f^n \in I$ implies $f \in I$ for all polynomials $f$ and $n \geq 0$, or in other word, $I$ is equal to its radical. 
Let $V$ be an algebraic set. A \emph{coordinate ring} of $V$ is the ring of polynomial functions from $V$ to $K$. 
A ring is called a \emph{computable ring} if its elements can be enumerated in a way that ring operations are computable functions.
\begin{fact}(\cite[Remarks at the end of §2 of Chapter 4]{CoxOShea2015} )\label{fact:radical-of-ideal-computable}
	Given ideal $I$ in $K[X]$, one can compute its radical.
\end{fact}
\begin{fact}[\cite{CoxOShea2015}, Chapter 4, §2, Theorem 7 (iii)]\label{fact:hilbert-nullstellensatz-cor}
	Let $K$ be an algebraically closed field. Then mappings $V \mapsto I(V), I \mapsto V(I)$ form a bijective correspondence between algebraic sets and radical ideals.
\end{fact}
\partofsec{Representation.}
For effectiveness of results, we assume that $K$ is a \emph{computable field}, i.e. its elements can be enumerated such that field operations are computable functions. Then we represent objects used by algorithms as follows:
\begin{itemize}
	\item ideals are represented by finite sets of generators; every ideal admits such a representation due to Hilbert's Basis Theorem (Fact \ref{fact:hilbert-basis-theorem}), and
	\item \emph{algebraic sets} are represented by ideals that they are zeroes of, i.e. an algebraic set $V$ is represented by $I$ such that $V = V(I)$.
\end{itemize}
\partofsec{Proof of Lemma \ref{lm:coordinate-ring-computable-from-V}.}
\begin{fact}\label{fact:I(V)-computable}
	Let $K$ be a computable, algebraically closed field. Given algebraic set $V$, one can compute $I(V)$.
\end{fact}
\begin{proof}
	 Let $V$ be represented by ideal $I$. Then $I(V)$ is the radical of $I$ (Fact \ref{fact:hilbert-nullstellensatz-cor}); it can be computed (Fact \ref{fact:radical-of-ideal-computable}).
\end{proof}
\begin{fact}(\cite[Chapter 5, §3, Proposition 5]{CoxOShea2015})\label{fact:quotient-rings-are-computable-rings}
	Quotient rings are computable rings. More precisely, given ideal $I$, ring $K[X]/I$ is a computable ring.
\end{fact}
\begin{fact}(\cite[Chapter 5, §2, remark about Theorem 7 after Definition 8]{CoxOShea2015})\label{fact:coord-ring-iso-to-KX/I}]
	For an algebraic set $V$, its coordinate ring is isomorphic to $K[X]/I(V)$.
\end{fact}
\begin{lemma}\label{lm:coordinate-ring-computable-from-V}
	Let $K$ be an algebraically closed field. Then, given algebraic set $V$, one can compute its coordinate ring and it is a computable ring.
\end{lemma}
\begin{proof}
	Coordinate ring is isomorphic to $K[X]/I(V)$ (Fact \ref{fact:coord-ring-iso-to-KX/I}). It can be computed from $V$, as we can compute $I(V)$ (Fact \ref{fact:I(V)-computable}). It is a computable ring (Fact \ref{fact:quotient-rings-are-computable-rings}).
\end{proof}
%